\documentclass[conference]{IEEEtran}

\usepackage{cite}
\usepackage{amsthm, amsmath,amssymb,amsfonts,bbm}
\usepackage{algorithmic}
\usepackage[dvips]{graphicx}
\usepackage{subfigure}
\graphicspath{{graphics/}}
\usepackage{multirow}
\usepackage{filecontents, pgfplots}
\usepackage[algoruled]{algorithm2e}
\usepackage{microtype}
\usepackage{tikz}
\usepackage{textcomp}
\usepackage{xcolor}
\usepackage{verbatim}
\usepackage[binary-units=true]{siunitx}
\DeclareSIUnit{\belmilliwatt}{Bm}
\DeclareSIUnit{\dBm}{\deci\belmilliwatt}
\def\BibTeX{{\rm B\kern-.05em{\sc i\kern-.025em b}\kern-.08em
		T\kern-.1667em\lower.7ex\hbox{E}\kern-.125emX}}

\makeatletter
\newif\iftag@here

\newcommand*{\taghere}[1][0pt]% #1 = additional vertical offset (optional)
{\ifmeasuring@\else% do not expand until displayed
	\global\tag@heretrue
	\tikz[remember picture,overlay]{\coordinate (taghere) at (0pt,#1);}%
	\fi}

\def\place@tag{%
	\iftagsleft@
	\kern-\tagshift@
	\iftag@here
	\global\tag@herefalse
	\tikz[remember picture,overlay]%
	{\path (taghere) -| node[anchor=base]{\rlap{\boxz@}} (0pt,0pt);}%
	\else
	\if1\shift@tag\row@\relax
	\rlap{\vbox{%
			\normalbaselines
			\boxz@
			\vbox to\lineht@{}%
			\raise@tag
	}}%
	\else
	\rlap{\boxz@}%
	\fi
	\kern\displaywidth@
	\fi% end of \iftag@here
	\else
	\kern-\tagshift@
	\iftag@here
	\global\tag@herefalse
	\tikz[remember picture,overlay]%
	{\path  (taghere) -|  node[anchor=base]{\llap{\boxz@}} (0pt,0pt);}%
	\else
	\if1\shift@tag\row@\relax
	\llap{\vtop{%
			\raise@tag
			\normalbaselines
			\setbox\@ne\null
			\dp\@ne\lineht@
			\box\@ne
			\boxz@
	}}%
	\else \llap{\boxz@}%
	\fi
	\fi% end of \iftas@here
	\fi
}
\makeatother

\DeclareMathOperator*{\argmax}{arg\,max}
\DeclareMathOperator*{\maximize}{maximize}
\DeclareMathOperator*{\argmin}{arg\,min}

\DeclareMathOperator*{\subjectto}{subject\,to}

\usepackage[acronym,nomain]{glossaries}
\makeglossaries
\newacronym{swipt}{SWIPT}{simultaneous wireless information and power transfer}
\newacronym{wpt}{WPT}{wireless power transfer}
\newacronym{wit}{WIT}{wireless information transfer}
\newacronym{iot}{IoT}{Internet-of-Things}
\newacronym{awgn}{AWGN}{additive white Gaussian noise}
\newacronym{tx}{TX}{transmitter}
\newacronym{ir}{IR}{information receiver}
\newacronym{eh}{EH}{energy harvester}
\newacronym{ap}{AP}{average power}
\newacronym{pp}{PP}{peak power}

\newacronym{siso}{SISO}{single-input single-output}
\newacronym{mimo}{MIMO}{multiple-input multiple-output}
\newacronym{miso}{MISO}{multiple-input single-output}
\newacronym{simo}{SIMO}{single-input multiple-output}

\newacronym{rf}{RF}{radio frequency}
\newacronym{dc}{DC}{direct current}
\newacronym{ac}{AC}{alternative current}
\newacronym{papr}{PAPR}{peak-to-average power ratio}
\newacronym{lp}{LPF}{low-pass filter}
\newacronym{mc}{MC}{matching circuit}
\newacronym{mrt}{MRT}{maximum ratio transmission}

\newacronym{rv}{RV}{random variable}
\newacronym{iid}{i.i.d.}{independent and identically distributed}
\newacronym{pdf}{pdf}{probability density function}
\glsunset{pdf}

\newacronym{dnn}{DNN}{dense neural networks}
\glsunset{dnn}
\newacronym{mdp}{MDP}{Markov decision process}

\newacronym{sca}{SCA}{successive convex approximation}
\newacronym{sdr}{SDR}{semi-definite relaxation}

\newacronym{spr}{LP}{low power}
\newacronym{mpr}{MP}{medium power}
\newacronym{lpr}{HP}{high power}

\DeclareMathOperator{\rank}{rank}
\newcommand{\norm}[1]{\left\lVert#1\right\rVert_2}
\newcommand{\Tr}[1]{\text{Tr}\{#1\} }
\begin{document}

	\newtheorem{proposition}{Proposition}	
	\newtheorem{lemma}{Lemma}	
	\newtheorem{corollary}{Corollary}
	\newtheorem{assumption}{Assumption}	
	\newtheorem{remark}{Remark}	
	
	\title{Harvested Power Region of Two-user MISO WPT Systems With Non-linear EH Nodes}
	\author{\IEEEauthorblockN{Nikita Shanin, Laura Cottatellucci, and Robert Schober}
	\IEEEauthorblockA{\textit{Friedrich-Alexander-Universit\"{a}t Erlangen-N\"{u}rnberg (FAU), Germany} }}	
	\maketitle
	
\begin{abstract}
In this paper, we determine the harvested power region of a two-user multiple-input single-output (MISO) wireless power transfer (WPT) system for a non-linear model of the rectennas at the energy harvester (EH) nodes.
To this end, we characterize the distributions of the transmit symbol vector that achieve individual points on the boundary of this region.
Each distribution is obtained as solution of an optimization problem where we maximize a weighted sum of the average harvested powers at the EH nodes under a constraint on the power budget of the transmitter.
We prove that the optimal transmit strategy employs two beamforming vectors and scalar unit norm transmit symbols with arbitrary phase.
To determine the beamforming vectors, we propose an iterative algorithm based on a two-dimensional grid search, semi-definite relaxation, and successive convex approximation.
Our numerical results reveal that the proposed design outperforms two baseline schemes based on a linear EH model and a single beamforming vector, respectively.
Finally, we observe that the harvested power region is convex and the power harvested at one EH node can be traded for a higher harvested power at the other node.
\end{abstract}
\setcounter{footnote}{0}
 
	\section{Introduction}
	\label{Section:Introduction}	
	The tremendous growth of the number of low-power \gls*{iot} devices over the past decade has exacerbated the problem of efficient charging of the batteries of these gadgets.
As a promising solution of this problem, far-field \gls*{wpt}, where power is transferred via \gls*{rf} signals, has attracted significant attention in recent years \cite{Grover2010, Zhang2013, Boshkovska2015, Kim2020, Clerckx2016a, Clerckx2018, Huang2017, Shen2020, Morsi2019, Shanin2020}.

For \gls*{siso} WPT systems with a single \gls*{eh} node, the authors of \cite{Grover2010} showed that the power delivered to the \gls*{eh} is maximized if a sinusoidal signal is broadcasted by the \gls*{tx}. 
The authors of \cite{Zhang2013} extended these results to \gls*{mimo} WPT systems and showed that the \emph{input power} at the EH is maximized if a scalar input symbol and \emph{energy beamforming} are employed at the \gls*{tx}.
However, practical EH circuits are non-linear \cite{Boshkovska2015, Kim2020} and, in this case, the optimality of the results in \cite{Grover2010} and \cite{Zhang2013} is severely degraded as far as the harvested power is concerned. 
Hence, an accurate modeling of the EH is crucial for the design of \gls*{wpt} systems \cite{Clerckx2016a, Clerckx2018, Huang2017, Shen2020, Morsi2019, Shanin2020}.

An EH node typically employs a rectenna, i.e., an antenna followed by a rectifier, which includes a non-linear diode.
In order to capture the non-linearities of the rectenna circuit, in \cite{Clerckx2016a}, the authors derived a non-linear EH model based on the Taylor series expansion of the current flow through the rectifying diode.
Based on this model, the authors in \cite{Clerckx2018} studied a WPT system with multiple antennas at the TX and a single antenna at the EH node, i.e., a single-user \gls*{miso} WPT system, and showed that the harvested power is maximized by energy beamforming \cite{Zhang2013}, which reduces to a scaled \gls*{mrt} in this case.
However, in \cite{Huang2017}, it was shown that energy beamforming is not optimal for WPT systems with multiple EH nodes employing non-linear rectenna circuits.
Furthermore, the authors in \cite{Huang2017} showed that, for multi-user \gls*{wpt} systems, there is a trade-off between the powers harvested by different EH nodes, which is characterized by a convex \emph{harvested power} region.
In \cite{Shen2020}, for a MIMO WPT system, where the EH node is equipped with multiple rectennas, the authors developed a framework for joint optimization of the transmit beamforming vector and the signal waveform for the maximization of the harvested power.

While the authors in \cite{Clerckx2016a, Clerckx2018, Huang2017, Shen2020} focused on the EH non-linearity in the low input power regime, the model developed in \cite{Morsi2019} also encompassed the saturation of the harvested power when the input power level at the EH is high.
Furthermore, the analysis in \cite{Morsi2019} and \cite{Shanin2020} showed that for \gls*{siso} WPT systems with a single EH node, it is optimal to adopt ON-OFF signaling at the TX, where the ON symbol and its probability are chosen to maximize the harvested power without saturating the EH while satisfying the average power constraint at the TX.
Although the results in \cite{Morsi2019} and \cite{Shanin2020} provide the optimal transmit strategy for single-user \gls*{siso} WPT systems, to the best of the authors' knowledge, the problem of optimizing the transmit symbol distribution for two-user \gls*{miso} WPT systems taking into account the EH circuit non-linearitites in both the low and high input power regimes has not been solved, yet. 
	
In this paper, we determine the achievable harvested power region for two-user \gls*{miso} WPT systems, where each EH node is equipped with a single rectenna.
We characterize the individual points on the boundary of the harvested power region by the corresponding probability density functions (pdfs) of the transmit symbol vector.
In order to take the non-linearity of the EH into account, we adopt the rectenna model derived in \cite{Morsi2019}.
Then, we formulate an optimization problem for the maximization of a weighted sum of the harvested powers at the EH nodes subject to a constraint on the power budget of the TX.
We show that the optimal transmit strategy employs scalar input symbols and two beamforming vectors.
To determine the two beamforming vectors and the corresponding probabilities, we develop an algorithm based on a two-dimensional grid search, \gls*{sdr}, and \gls*{sca} \cite{Coope2001, Ghanem2020, Sun2017}.
Our simulation results reveal that the proposed two-user WPT design with two beamforming vectors outperforms baseline schemes based on the linear EH model in \cite{Zhang2013} and a single beamforming vector at the TX, respectively.
Furthermore, we observe that by adjusting the user weights, the TX can control the distribution of the harvested power among the individual EH nodes.

The remainder of this paper is organized as follows. 
In Section II, we introduce the system model. % and the adopted EH model.
In Section III, we formulate an optimization problem and develop an algorithm for obtaining a solution of the problem.
In Section IV, we provide numerical results to evaluate the performance of the proposed design.
Finally, in Section V, we draw some conclusions.

\emph{Notation:} Bold upper case letters $\boldsymbol{X}$ represent matrices and ${X}_{i,j}$ denotes the element of $\boldsymbol{X}$ in row $i$ and column $j$. 
Bold lower case letters $\boldsymbol{x}$ stand for vectors and ${x}_{i}$ is the $i^\text{th}$ element of $\boldsymbol{x}$.
$\boldsymbol{X}^H$, $\Tr{\boldsymbol{X}}$, and $\rank \{\boldsymbol{X}\}$ denote the Hermitian transpose, trace, and rank of matrix $\boldsymbol{X}$, respectively.
The expectation with respect to \gls*{rv} $x$ is denoted by $\mathbb{E}_x\{\cdot\}$. 
The real part of a complex number is denoted by $\Re\{ \cdot \}$.
$\norm{\boldsymbol{x}}$ represents the L2-norm of $\boldsymbol{x}$.
The imaginary unit is denoted by $j$.
$\boldsymbol{0}_{N}$ denotes the square zero matrix of size $N \times N$.
The sets of real and complex numbers are denoted by $\mathbb{R}$ and $\mathbb{C}$, respectively.
The Dirac delta function is denoted by $\delta(x)$.

	\section{System Model}
	\label{Section:SystemModelPreliminaries}
	In this section, we present the considered two-user MISO WPT system and discuss the adopted EH model.
		\subsection{MISO WPT System Model}
		\label{Section:SystemModel}
		We consider a narrow-band two-user\footnotemark  MISO \gls*{wpt} system comprising a \gls*{tx} with $N_\text{t} \geq 1$ antennas and two EH nodes, where each EH node is equipped with a single antenna.
The TX broadcasts a pulse-modulated RF signal, whose equivalent complex baseband (ECB) representation is modeled as $\boldsymbol{x}(t) = \boldsymbol{x}[n] \psi(t-nT)$, where $\boldsymbol{x}[n] \in \mathbb{C}^{N_\text{t}}$ is the transmitted vector in time slot $n$, $\psi(t)$ is a rectangular transmit pulse, and $T$ is the symbol duration.
The transmit vectors $\boldsymbol{x}[n]$ are mutually independent realizations of a random vector $\boldsymbol{x}$, whose \gls*{pdf} is denoted by $p_{\boldsymbol{x}}(\boldsymbol{x})$.
\footnotetext{Although, in this paper, for simplicity of presentation, we consider two-user MISO WPT systems, the extension of the proposed framework to the multi-user case is possible.}

The ECB channel between the TX and EH node $m, m\in\{1,2\},$ is characterized by row-vector $\boldsymbol{g}_m \in \mathbb{C}^{1 \times N_\text{t}}$.
Thus, the RF signal received in time slot $n$ at EH node $m$ is given by ${z^\text{RF}_{m}}(t) = \sqrt{2} \Re\{ \boldsymbol{g}_m \boldsymbol{x}(t) \exp(j 2 \pi f_c t) \} $, where $f_c$ denotes the carrier frequency.
The noise received at the EH nodes is ignored since its contribution to the harvested energy is negligible.

		\subsection{Energy Harvester Model}
		\label{Section:EhModel}
		
As in \cite{Morsi2019}, we assume that each EH node is equipped with a rectenna.
Each rectenna comprises an antenna, a \gls*{mc}, a non-linear rectifier with a low-pass filter, and a load resistor \cite{Morsi2019, Shanin2020}. 
In order to maximize the power delivered to the rectifier, the MC is typically well-tuned to the carrier frequency $f_c$ and is designed to match the input impedance of the non-linear rectifier circuit with the output impedance of the antenna \cite{Morsi2019}.
The rectifier is an electrical circuit that comprises a non-linear diode and a low-pass filter to convert the RF signal ${z^\text{RF}_{m}}(t)$ received in time slot $n$ by EH $m$ to a \gls*{dc} signal at the load resistor $R_\text{L}$ of the rectenna.

In this paper, we adopt the EH model derived in \cite{Morsi2019}. 
The corresponding power harvested by the rectenna as a function of the magnitude of the received ECB signal, $z_m[n] = \boldsymbol{g}_m \boldsymbol{x}[n]$, is given by\footnote{In this paper, as in \cite{Morsi2019}, we assume that both EH nodes are memoryless and have identical parameters. Hence, we drop time slot index $n$ and EH node index $m$.} 
\begin{equation}
	\phi(|{z}|^2) = \min\big\{ \varphi(|z|^2), \varphi(A_s^2) \big\},
	\label{Eqn:RaniaModel}
\end{equation}
where $\varphi(|z|^2) = \bigg[\frac{1}{a} W_0 \bigg(a\exp(a) I_0 \Big(B\sqrt{2|{z}|^2}\Big) \bigg)-1 \bigg]^2 I_s^2 R_L$, $a = \frac{I_s(R_L + R_s)}{\mu V_\text{T}}$, $B = [\mu V_\text{T} \sqrt{\Re \{1 / Z_a^* \}}]^{-1}$, and $W_0(\cdot)$ and $I_0(\cdot)$ are the principal branch of the Lambert-W function and the zeroth order modified Bessel function of the first kind, respectively.
Here, $Z_a^*$, $V_\text{T}$, $I_s$, $R_s$, and $\mu \in [1,2]$ are parameters of the rectenna circuit, namely, the complex-conjugate of the input impedance of the rectifier circuit, the thermal voltage, the reverse bias saturation current, the series resistance, and the ideality factor of the diode, respectively. 
These parameters depend on the circuit elements and are independent of the received signal.
Finally, since for large input power levels, rectenna circuits are driven into saturation \cite{Boshkovska2015, Morsi2019, Shanin2020}, the function in (\ref{Eqn:RaniaModel}) is bounded, i.e., $\phi(|z|^2) \leq \phi(A_s^2) , \; \forall z \in \mathbb{C}$, where $A_s$ is the minimum input signal magnitude, for which the output power starts to saturate.

		\section{Problem Formulation and Solution}
		\label{Section:ProblemFormulation}
		In this section, we formulate an optimization problem for the maximization of the weighted sum  of the harvested powers of the considered two-user WPT system.
		In order to solve this problem, we formulate and solve an auxiliary one-dimensional optimization problem, where we maximize the expectation of a function of a \gls*{rv} under a constraint on its mean value.
		Finally, we characterize the optimal random transmit symbol vectors as solution of the original optimization problem and develop an algorithm to obtain a suboptimal solution.
		\subsection{Problem Formulation}
		We characterize individual points in the harvested power region by the pdf $p_{\boldsymbol{x}}(\boldsymbol{x})$ of transmit symbol vector $\boldsymbol{x}$.
The boundary of this region is obtained by considering all the possible convex combinations of the weighted harvested powers at the EH nodes \cite{Huang2017}.
Hence, in order to determine the individual points on the boundary of the harvested power region, we maximize the weighted average power harvested at the EH nodes under an average power constraint at the TX.
Thus, we formulate the following optimization problem\footnotemark:
\begin{subequations}
	\begin{align}
	\maximize_{ {p}_{\boldsymbol{x}} } \quad &\overline{\Phi}(p_{\boldsymbol{x}}) 
	\label{Eqn:WPT_GeneralProblem_Obj} \\
	\subjectto \quad & \int_{\boldsymbol{x}} \norm{\boldsymbol{x}}^2 p_{\boldsymbol{x}}(\boldsymbol{x}) d \boldsymbol{x} \leq P_x, \label{Eqn:WPT_GeneralProblem_C1}\\
	&\int_{\boldsymbol{x}} p_{\boldsymbol{x}}(\boldsymbol{x}) d \boldsymbol{x} = 1. \label{Eqn:WPT_GeneralProblem_C2}
	\end{align}
	\label{Eqn:WPT_GeneralProblem}
\end{subequations}
\footnotetext{We note that the solution of (\ref{Eqn:WPT_GeneralProblem}) may not be unique. 
In particular, if affordable by the power budget $P_x$, there may be an infinite number of pdfs $p_{\boldsymbol{x}}(\boldsymbol{x})$ that drive both EH nodes into saturation, i.e., yield the maximum value of $\overline{\Phi}(p_{\boldsymbol{x}}) = \phi(A_s^2)$, while satisfying the constraints (\ref{Eqn:WPT_GeneralProblem_C1}), (\ref{Eqn:WPT_GeneralProblem_C2}).
Since, in this paper, we aim at characterizing the harvested power region, in the following, we determine one pdf $p^*_{\boldsymbol{x}}$ that solves (\ref{Eqn:WPT_GeneralProblem}).}
Here, the weighted sum of the average harvested powers at the EH nodes is defined as
\begin{equation}
\overline{\Phi}(p_{\boldsymbol{x}}) = \sum_{m = 1}^2 \xi_m \mathbb{E}_{\boldsymbol{x}} \big\{ \phi(|\boldsymbol{g}_m \boldsymbol{x}|^2) \big\},
\label{Eqn:EhUtilityFunction}
\end{equation}
where $\xi_m \geq 0, m\in\{1,2\}, \sum_m \xi_m = 1,$ is the weight for EH node $m$ \cite{Shen2020}.
We note that the weights associated with the users allow the TX to control the distribution of the harvested power among the EH nodes.
In particular, if weight $\xi_m$ is increased, the average harvested power at EH node $m$ will be increased at the expense of the average harvested power at the other EH node.
Furthermore, we impose constraints (\ref{Eqn:WPT_GeneralProblem_C1}) and (\ref{Eqn:WPT_GeneralProblem_C2}) to limit the transmit power budget at the TX and to ensure that $p_{\boldsymbol{x}}(\boldsymbol{x})$ is a valid pdf, respectively.

In order to find the optimal solution of (\ref{Eqn:WPT_GeneralProblem}), it is convenient to first solve a related auxiliary optimization problem.
To this end, in the next subsection, we consider the maximization of the expectation of a non-decreasing function $f(\nu)$ of a scalar \gls*{rv} $\nu$ under a constraint on the mean value of $\nu$.

		\subsection{Auxiliary Optimization Problem}
		Let us consider the following auxiliary optimization problem:
\begin{equation}
	\maximize_{ {p}_{\nu} } \; \mathbb{E}_\nu \{f(\nu)\} \quad \subjectto \; \mathbb{E}_\nu \{\nu\} \leq \overline{\nu},
	\label{Eqn:GeneralOptimizationProblem}
\end{equation}
\noindent where we optimize the pdf $p_\nu(\nu)$ of $\nu$ for the maximization of the expectation of $f(\nu)$ under a constraint $\overline{\nu}$ on the mean value of $\nu$.
In order to solve (\ref{Eqn:GeneralOptimizationProblem}), let us first define the slope of the straight line connecting points $(\nu_1, f(\nu_1))$ and $(\nu_2, f(\nu_2))$, where $\nu_2 > \nu_1$, as follows:
\begin{equation}
	s(\nu_1,\nu_2) = \frac{f(\nu_2) - f(\nu_1)}{\nu_2 - \nu_1}.
	\label{Eqn:SlopeFunctionF}
\end{equation}
We note that if $f(\nu)$ is convex (concave), the solution of (\ref{Eqn:GeneralOptimizationProblem}) is given by the Edmundson-Madansky (Jensen's) inequality, see, e.g., \cite{Dokov2002}. 
However, since function $\phi(\cdot)$ in (\ref{Eqn:RaniaModel}) is neither convex nor concave, in the following lemma, we extend the results in \cite{Dokov2002} to arbitrary non-decreasing functions $f(\nu)$ and determine the solution of (\ref{Eqn:GeneralOptimizationProblem}).
\begin{lemma}
	Let us consider a one-dimensional non-decreasing function $f(\nu)$ of a \gls*{rv} $\nu$.
	A solution of optimization problem (\ref{Eqn:GeneralOptimizationProblem}) is a discrete pdf given by $p_\nu^*(\nu) = \beta \delta(\nu - \nu_1^*) + (1-\beta) \delta(\nu - \nu_2^*)$, where
	$\nu_1^* = \argmin_{\nu_1 \leq \overline{\nu} } c(\nu_1)$, $c(\nu_1) = \max_{\nu_2 \geq \overline{\nu}} s(\nu_1,\nu_2)$, $\nu_2^* = \argmax_{\nu_2 \geq \overline{\nu}} s(\nu_1^*,\nu_2)$, and $\beta = \frac{\nu_2^* - \overline{\nu}}{\nu_2^* - \nu_1^*}$.
	\label{Theorem:Corollary2}	
\end{lemma}
\begin{proof}
	Please refer to Appendix~\ref{Appendix:LemmaProof}.
\end{proof}

We note that similar to (\ref{Eqn:WPT_GeneralProblem}), the solution of the optimization problem (\ref{Eqn:GeneralOptimizationProblem}) may not be unique.
Lemma~\ref{Theorem:Corollary2} reveals that for the maximization of the expectation of a non-decreasing function $f(\nu)$, the pdf $p_\nu(\nu)$ is discrete and has two mass points, $\nu_1^*$ and $\nu_2^*$, which are obtained as solutions of a min-max optimization problem. 
In the following, we exploit the result in Lemma~\ref{Theorem:Corollary2} for solving the optimization problem in (\ref{Eqn:WPT_GeneralProblem}).

		\subsection{Solution of Problem (\ref{Eqn:WPT_GeneralProblem})}
		\label{Section:MIMOSystem}
		In the following, we solve optimization problem (\ref{Eqn:WPT_GeneralProblem}). 
In Proposition~\ref{Theorem:Proposition3}, we determine the optimal transmit symbol vectors $\boldsymbol{x}$ that characterize the optimal pdf $p^*_{\boldsymbol{x}}(\boldsymbol{x})$ as solution of (\ref{Eqn:WPT_GeneralProblem}).
\begin{proposition}
	For the considered two-user MISO WPT systems, function $\overline{\Phi} (\cdot)$ is maximized for transmit symbol vectors $\boldsymbol{x} = \boldsymbol{w} s$, where $s$ are unit norm symbols  $s = \exp(j \phi_s)$ with arbitrary phases $\phi_s$.
	Here, $\boldsymbol{w}$ is a discrete random beamforming vector with $p^*_{\boldsymbol{w}}(\boldsymbol{w}) = \beta \delta(\boldsymbol{w} - \boldsymbol{w}^*_1) + (1-\beta) \delta(\boldsymbol{w} - \boldsymbol{w}^*_2)$, where $\beta = \frac{\norm{\boldsymbol{w}_2^*}^2 - P_x}{\norm{\boldsymbol{w}_1^*}^2 - \norm{\boldsymbol{w}_2^*}^2}$.
	The beamforming vectors $\boldsymbol{w}^*_n, \; n\in\{1,2\}$, are given by
	\begin{align}
	\boldsymbol{w}^*_n &= \{ \boldsymbol{w} : \Psi(\boldsymbol{w}) = \Phi(\nu_n^*) \},
	\label{Eqn:MimoPropositionBeamformerProblem} \\
	\Phi(\nu) &= \underset{ \{\boldsymbol{w} \; | \; \| \boldsymbol{w} \|_2^2  = \nu, \; \boldsymbol{w} \in \mathbb{C}^{N_t} \} }{\max}\Psi(\boldsymbol{w}),
	\label{Eqn:MimoPropositionFunction}
	\end{align}
	where $\Psi(\boldsymbol{x}) = \sum_m \xi_m \phi(|\boldsymbol{g}_m \boldsymbol{x}|^2).$ 
	Here, $\nu^*_1$ and $\nu^*_2$ are defined as in Lemma~\ref{Theorem:Corollary2} with $\overline{\nu} = P_x$ and $s(\nu_1, \nu_2) = \big(\Phi(\nu_2) - \Phi(\nu_1)\big) / \big(\nu_2 - \nu_1\big)$.
	\label{Theorem:Proposition3}
\end{proposition}
\begin{proof}
	Please refer to Appendix~\ref{Appendix:Prop3Proof}.
\end{proof}

Proposition~\ref{Theorem:Proposition3} reveals that the optimal transmit vectors $\boldsymbol{x}$ are characterized by scalar unit norm symbols $s$ with arbitrary phases\footnotemark and two beamforming vectors, $\boldsymbol{w}^*_1$ and $\boldsymbol{w}^*_2$, respectively.
\footnotetext{We note that the phase $\phi_s$ of scalar symbol $s$ can be chosen arbitrarily in each time slot $n$. This degree of freedom can be further exploited, for example, to transmit information \cite{Shanin2020}.}
We note that in order to determine the optimal vectors $\boldsymbol{w}^*_n, n\in\{1,2\},$ the function $\Phi(\cdot)$ and the values $\nu_1^*$, $\nu_2^*$ as solutions of (\ref{Eqn:MimoPropositionFunction}) and the non-convex min-max optimization problems in Lemma~\ref{Theorem:Corollary2}, respectively, are required.
In the following, in order to determine the values of function $\Phi(\cdot)$, we first exploit \gls*{sdr} and \gls*{sca} and develop a low-complexity algorithm to determine a suboptimal solution of (\ref{Eqn:MimoPropositionFunction}).
Then, due to the low dimensionality of the min-max problem, we propose to obtain the optimal values $\nu_1^*$, $\nu_2^*$ and the corresponding vectors $\boldsymbol{w}^*_n, n\in\{1,2\},$ via a two-dimensional grid search \cite{Coope2001}.

			\subsubsection{Solution of (\ref{Eqn:MimoPropositionFunction})}
			\label{Section:SuboptimalSolution}
			Since (\ref{Eqn:MimoPropositionFunction}) is a non-convex optimization problem, determining its optimal solution is, in general, NP-hard.
Therefore, in the following, we develop an iterative algorithm to obtain a suboptimal solution of (\ref{Eqn:MimoPropositionFunction}).
To this end, we first define matrix $\boldsymbol{W} = \boldsymbol{w} \boldsymbol{w}^H$ and reformulate problem (\ref{Eqn:MimoPropositionFunction}) equivalently as follows:
\begin{subequations}
	\begin{align}
	\maximize_{\boldsymbol{W} \in \mathcal{S}_{+}} \quad & \psi(\boldsymbol{W}) 
	\label{Eqn:MimoSuboptimalFunctionRef_Obj} \\
	\subjectto \quad &\Tr{\boldsymbol{W}} \leq \nu, \label{Eqn:MimoSuboptimalFunctionRef_C1} \\
	&\rank \{ \boldsymbol{W} \} = 1, \label{Eqn:MimoSuboptimalFunctionRef_C2}
	\end{align}	
	\label{Eqn:MimoSuboptimalFunctionRef}
\end{subequations}
\noindent\hspace*{-3.3pt}where $\psi(\boldsymbol{W}) = \sum_{m = 1}^{2} \xi_m \phi(\boldsymbol{g}_m \boldsymbol{W} \boldsymbol{g}_m^H)$ and $\mathcal{S}_{+}$ denotes the set of positive semidefinite matrices.
Since the objective function in (\ref{Eqn:MimoPropositionFunction}) is monotonic non-decreasing in $|\boldsymbol{w}|$, we equivalently replace the equality constraint in (\ref{Eqn:MimoPropositionFunction}) by inequality constraint (\ref{Eqn:MimoSuboptimalFunctionRef_C1}).%\cite{Zhang2012}.

Optimization problem (\ref{Eqn:MimoSuboptimalFunctionRef}) is non-convex due to the non-convexity of the objective function (\ref{Eqn:MimoSuboptimalFunctionRef_Obj}) and constraint (\ref{Eqn:MimoSuboptimalFunctionRef_C2}). 
Therefore, in order to obtain a suboptimal solution of (\ref{Eqn:MimoSuboptimalFunctionRef}), we employ \gls*{sdr} and drop constraint (\ref{Eqn:MimoSuboptimalFunctionRef_C2}).
Next, we define sets $\mathcal{W}_{i,j}, i,j \in \{0,1\},$ where $\forall \, \boldsymbol{W} \in \mathcal{W}_{i,j}$, EH node $m=1$ and $m=2$ is driven into saturation if and only if $i=1$ and $j=1$, respectively.
We note that $\mathcal{W}_{0,0} \cup \mathcal{W}_{0,1} \cup \mathcal{W}_{1,0} \cup \mathcal{W}_{1,1} = \mathcal{S}_{+}$.
Since EH node $m$ is driven into saturation if and only if $\boldsymbol{g}_{m} \boldsymbol{W} \boldsymbol{g}_{m}^H \geq A_s^2$, cf. (\ref{Eqn:RaniaModel}), subset $\mathcal{W}_{i,j},  i,j\in\{0,1\},$ is convex and given by $\mathcal{W}_{i,j} = \big\{ \boldsymbol{W}: (-1)^i (\boldsymbol{g}_{1} \boldsymbol{W} \boldsymbol{g}_{1}^H - A_s^2) \leq 0, \; (-1)^j (\boldsymbol{g}_{2} \boldsymbol{W} \boldsymbol{g}_{2}^H - A_s^2) \leq 0, \boldsymbol{W} \in \mathcal{S}_{+}  \big\}$.
Furthermore, since function $\phi(|z|^2)$ is convex in the intervals $|z|^2 \in [0, A_s^2)$ and $|z|^2 \in [A_s^2, +\infty)$, as a weighted sum of convex functions, objective function $\psi(\boldsymbol{W})$ is convex in each subset $\mathcal{W}_{i,j}, i,j \in \{0,1\}$.
Therefore, we split the relaxed non-convex optimization problem (\ref{Eqn:MimoSuboptimalFunctionRef_Obj}), (\ref{Eqn:MimoSuboptimalFunctionRef_C1}) into four problems as follows:
\begin{equation}
	\maximize_{\boldsymbol{W} \in \mathcal{W}_{i,j} } \quad \psi(\boldsymbol{W}) \;
	\subjectto  \, (\text{\upshape\ref{Eqn:MimoSuboptimalFunctionRef_C1}}),
	\label{Eqn:MimoSuboptimalFunctionRefConv}
\end{equation}	
\noindent for $i,j \in \{0,1\}$.
Then, as a suboptimal solution of (\ref{Eqn:MimoSuboptimalFunctionRef}), we adopt 
\begin{equation}
	\boldsymbol{W}^* = \argmax \big\{ \psi(\boldsymbol{W}^*_{i,j}), \; i,j\in\{0,1\} \big\},
	\label{Eqn:MimoSuboptimalSolution}
\end{equation} 
\noindent where $\boldsymbol{W}^*_{i,j}$ is the solution of (\ref{Eqn:MimoSuboptimalFunctionRefConv}) for $i,j \in \{0,1\}$.

Next, in order to solve (\ref{Eqn:MimoSuboptimalFunctionRefConv}), we employ \gls*{sca} \cite{Sun2017}.
To this end, we construct an underestimate of the convex objective function $\psi(\boldsymbol{W})$ as follows:
\begin{equation}
{\psi}(\boldsymbol{W}) \geq \hat{\psi}(\boldsymbol{W}, \boldsymbol{W}^{(k)}),
\end{equation}
\noindent where $\boldsymbol{W}^{(k)}$ is the solution obtained in iteration $k$ of the \gls*{sca} algorithm, $\hat{\psi}(\boldsymbol{W}, \boldsymbol{W}^{(k)}) = \psi(\boldsymbol{W}^{(k)}) + \Tr{ \triangledown{\psi}(\boldsymbol{W}^{(k)})^H (\boldsymbol{W} - \boldsymbol{W}^{(k)} ) }$, and $\triangledown{\psi}(\boldsymbol{W}^{(k)})$ is the gradient of $\psi(\cdot)$ evaluated at $\boldsymbol{W}^{(k)}$ \cite{Ghanem2020}.
Thus, in each iteration $k$ of the proposed algorithm, we solve the following optimization problem:
\begin{equation}
\begin{aligned}
\boldsymbol{W}^{(k+1)} = \argmax_{\boldsymbol{W} \in  \mathcal{W}_{i,j}  } \, \hat{\psi}(\boldsymbol{W}, \boldsymbol{W}^{(k)})  \;
\subjectto \,  (\text{\upshape\ref{Eqn:MimoSuboptimalFunctionRef_C1}}).	
\end{aligned}	
\label{Eqn:MimoSuboptimalFunctionAlg}
\end{equation}
\noindent We note that if (\ref{Eqn:MimoSuboptimalFunctionAlg}) is feasible, the optimization problem is convex and can be solved with standard numerical optimization tools, such as CVX \cite{Grant2015}.
Since $\psi(\boldsymbol{W}^{(k+1)}) \geq \hat{\psi}(\boldsymbol{W}^{(k+1)}, \boldsymbol{W}^{(k)}) \geq \psi(\boldsymbol{W}^{(k)})$, the proposed algorithm converges to a limit point of (\ref{Eqn:MimoSuboptimalFunctionRefConv}).
Finally, we obtain the beamforming vector $\boldsymbol{w}^*$ as the dominant eigenvector\footnotemark of the solution $\boldsymbol{W}^*$ of (\ref{Eqn:MimoSuboptimalSolution}) and compute the corresponding value of function $\Phi(\nu)~=~\Psi(\boldsymbol{w}^*)$.
\footnotetext{It can be shown that for $\nu > 0$, the optimal solution $\boldsymbol{W}^*$ meets $\rank \{\boldsymbol{W^*}\} = 1$. The corresponding proof is similar to the one in \cite[Appendix]{Xu2019} and is omitted here due to the space constraints. }
The proposed algorithm is summarized in Algorithm~\ref{OptimizationAlgorithmFunction}.
The computational complexity of the proposed algorithm per inner loop iteration scales with $N_\text{t}$ and is given by\footnotemark $\mathcal{O}\big(N_\text{t}^{\frac{7}{2}}\big)$, where $\mathcal{O}(\cdot)$ is the big-O notation.
\footnotetext{The computational complexity of a semidefinite optimization problem involving $m$ constraints and an $n \times n$ positive semidefinite matrix is given by $\mathcal{O}\big(\sqrt{n}(mn^3 + m^2n^2 + m^3) \big)$ \cite{Polik2010}. For (\ref{Eqn:MimoSuboptimalFunctionAlg}), we have $m = 3$ and $n = N_\text{t}$.}
\begin{algorithm}[tb]				
	\SetAlgoNoLine%		
	\SetKwFor{Foreach}{for each}{do}{end}		
	Initialize: Transmit power $\nu$, tolerance error $\epsilon_\text{SCA}$.	\\	
	1. Set iteration indices $i = 0, j = 0$, initialize matrices $\boldsymbol{W}^*_{i,j} = \boldsymbol{0}_{N_\text{t}}, i,j\in\{0,1\}$.\\
	\For{$i = 0$ {\upshape to} $1$}{	
		\For{$j = 0$ {\upshape to} $1$}{		
			2.1 Set initial values $h^{(0)} = 0$, $k = 0$, and randomly initialize $\boldsymbol{W}^{(0)} \in \mathcal{W}_{i,j}$ \\
			\If{\text{\upshape (\ref{Eqn:MimoSuboptimalFunctionAlg}) is feasible} }{
				\Repeat{$|h^{(k)}-h^{(k-1)}|\leq \epsilon_\text{\upshape{SCA}}$ }{		
					a. For given $\boldsymbol{W}^{(k)}$, obtain $\boldsymbol{W}^{(k+1)}$ as solution of (\ref{Eqn:MimoSuboptimalFunctionAlg}) \\								
					b. Evaluate $h^{(k+1)} = {\psi}(\boldsymbol{W}^{(k+1)})$\\
					c. Set $k = k+1$\\ 
				}				
			2.2. Set $\boldsymbol{W}^*_{i,j} = \boldsymbol{W}^{(k)}$
			}					
		}	
	}
	3. Determine $\boldsymbol{W}^*$ as solution of (\ref{Eqn:MimoSuboptimalSolution})\\
	4. Obtain $\boldsymbol{w}^*$ as the dominant eigenvector of $\boldsymbol{W}^{*}$ and evaluate $\Phi(\nu) = \Psi(\boldsymbol{w}^*)$\\
	\textbf{Output:} 
	$ \Phi(\nu)$, $\boldsymbol{w}^*$
	\caption{\strut Suboptimal algorithm for solving optimization problem (\ref{Eqn:MimoPropositionFunction}) }
	\label{OptimizationAlgorithmFunction}
\end{algorithm}

			\subsubsection{Grid Search Method}
			\label{Section:Gridearch}
			In the following, we propose a grid-search based method for solving the min-max optimization problem in Proposition~\ref{Theorem:Proposition3}.
We note that this problem is not convex since function $s(\nu_1, \nu_2)$ is not convex and not concave in $\nu_1$ and $\nu_2$, respectively. 
However, since the dimensionality of the problem is low, performing a grid search to determine $\nu_1^*$ and $\nu_2^*$ is feasible \cite{Coope2001}.
To this end, we define a uniform grid $\mathcal{P} = \{\rho_0, \rho_1, \rho_2, ..., \rho_{N_\rho}\}$, where $\rho_0 = 0$, $\rho_j = \Delta_\rho + \rho_{j-1}$, $j = 1,2,...,{N_\rho}$, and $\Delta_\rho$ is a predefined step size.
$N_\rho$ is the grid size that is chosen sufficiently large in our simulations to ensure that the function $\Phi(\cdot)$ saturates, i.e., both EH nodes are driven into saturation for $\nu = \rho_{N_\rho}$
Then, we define the smallest element of $\mathcal{P}$, which is larger than $P_x$ as $\rho_n$, i.e., $\rho_n = \min\{\rho_j | \rho_j \geq P_x, j = 0,1,...,N_\rho\}$.
Next, we define a matrix $\boldsymbol{S} \in \mathbb{R}^{n\times (N_\rho-n+1)}$, whose elements are the values of function $s(\cdot, \cdot)$ evaluated at the points of $\mathcal{P}$, i.e., ${S}_{i,j} = s(\rho_i, \rho_{j'})$, $i = 0,1,...,{n-1}$, $j = {j'}-n$, and ${j'} = n,n+1,...,N_\rho$.
Finally, we obtain the power values $\nu_1^* = \rho_{i^*}$ and $\nu_2^*=\rho_{n+j^*}$, where $i^* = \argmin_{i} \max_{j} {S}_{i,j}$ and $j^* = \argmax_{j} {S}_{i^*,j}$, respectively, and the corresponding beamforming vectors $\boldsymbol{w}_n^*, n \in \{1,2\},$ in (\ref{Eqn:MimoPropositionBeamformerProblem}).
The proposed scheme is summarized in Algorithm~\ref{AlgorithmGridSearch}.
The computational complexity of the proposed scheme is quadratic with respect to the grid size $N_\rho$ and does not depend on the number of transmit antennas.
\begin{algorithm}[!t]				
	\SetAlgoNoLine%		
	\SetKwFor{Foreach}{for each}{do}{end}		
	Initialize: Grid size $N_{\rho}$, step size $\Delta_{\rho}$, maximum TX power $P_x$, initial values $\rho_0=0$, $m=0$.	\\	
	1. Compute the grid $\mathcal{P}$ and the values of $\Phi(\cdot)$ for the grid elements:\\
	\Repeat{$m > N_\rho$ }{
		1.1. Compute $\Phi_{m}$ = $\Phi(\rho_{m})$ and $\boldsymbol{v}_m = \boldsymbol{w}^*$ with Algorithm~\ref{OptimizationAlgorithmFunction}\\
		1.2. Set $\rho_{m+1} = \rho_{m} + \Delta_{\rho}$ \\
		1.3. $m = m+1$\\
	}
	2. Determine the grid element $\rho_n = \min\{\rho_j | \rho_j \geq P_x, j = 0,1,...,N_\rho \}$ \\
	3. Calculate the elements of matrix $\boldsymbol{S}$ as ${S}_{i,j} = s(\rho_i, \rho_{j'}) = \frac{\Phi_{j'} - \Phi_i}{\rho_{j'} - \rho_i}$, $i = 0,1,...,{n-1}$, $j = {j'}-n$, and ${j'} = n,n+1,...,N_\rho$ \\
	4. Determine power values and the corresponding vectors $\nu_1^* = \rho_{i^*}, \boldsymbol{w}^*_1 = \boldsymbol{v}_{i^*}$ and $\nu_2^*=\rho_{n+j^*}, \boldsymbol{w}^*_2 = \boldsymbol{v}_{n+j^*}$, where $i^* = \argmin_{i} \max_{j} {S}_{i,j}$ and $j^* = \argmax_{j} {S}_{i^*,j}$ \\
	\textbf{Output:} Optimal vectors $\boldsymbol{w}_1^*$, $\boldsymbol{w}_2^*$, and $\beta = \frac{\nu_2^* - P_x}{\nu_2^*-\nu_1^*}$
	\caption{\strut Grid search for determining the optimal vectors $\boldsymbol{w}_1^*$ and $\boldsymbol{w}_2^*$ }
	\label{AlgorithmGridSearch}
\end{algorithm}

	\section{Numerical Results}
	In this section, we study the harvested power regions via simulations.
		In our simulations, the path losses are calculated as $35.3+37.6\log_{10}(d_m)$, where $d_m$ is the distance between the TX and EH node $m$ \cite{Ghanem2020}.
Furthermore, in order to harvest a meaningful amount of power, we assume that the TX and each EH node have a line-of-sight link.
Thus, the channel gains $\boldsymbol{g}_m$ follow Rician distributions with Rician factor $1$.
For the EH model in (\ref{Eqn:RaniaModel}), we adopt the parameter values {$a=1.29$, $B = 1.55\cdot10^3$, $I_s = \SI{5}{\micro\ampere}$, $R_L = \SI{10}{\kilo\ohm}$, and $A_s^2 = \SI{25}{\micro\watt}$}, respectively \cite{Morsi2019}.
For Algorithms~\ref{OptimizationAlgorithmFunction} and \ref{AlgorithmGridSearch}, we adopt a step size, grid size, and error tolerance value of $\Delta_{\rho} = 0.1$, $N_{\rho} = 10^3$, and $\epsilon_{\text{SCA}} = 10^{-3}$.
All simulation results were averaged over $100$ channel realizations.

		\begin{figure*}[!t]
	\centering
	\subfigure[Low transmit power regime, $P_x = \SI{5}{\watt}$]{
		\includegraphics[width=0.35\textwidth, draft=false]{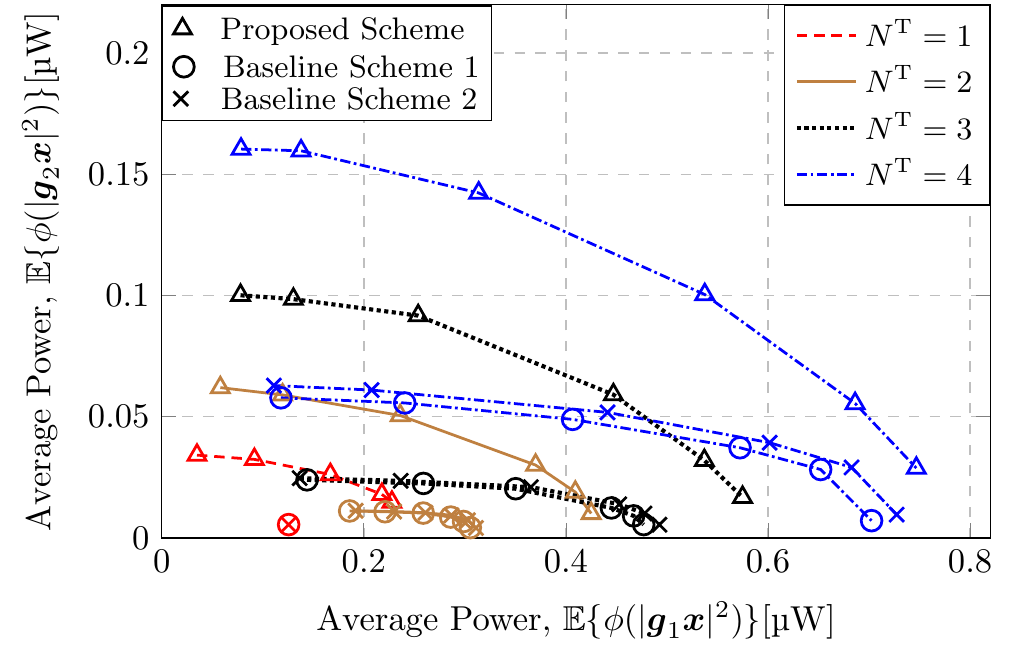} \label{Fig:Results_MU_Region_LP}}
	\quad
	\subfigure[High transmit power regime, $P_x = \SI{30}{\watt}$]{
		\includegraphics[width=0.35\textwidth, draft=false]{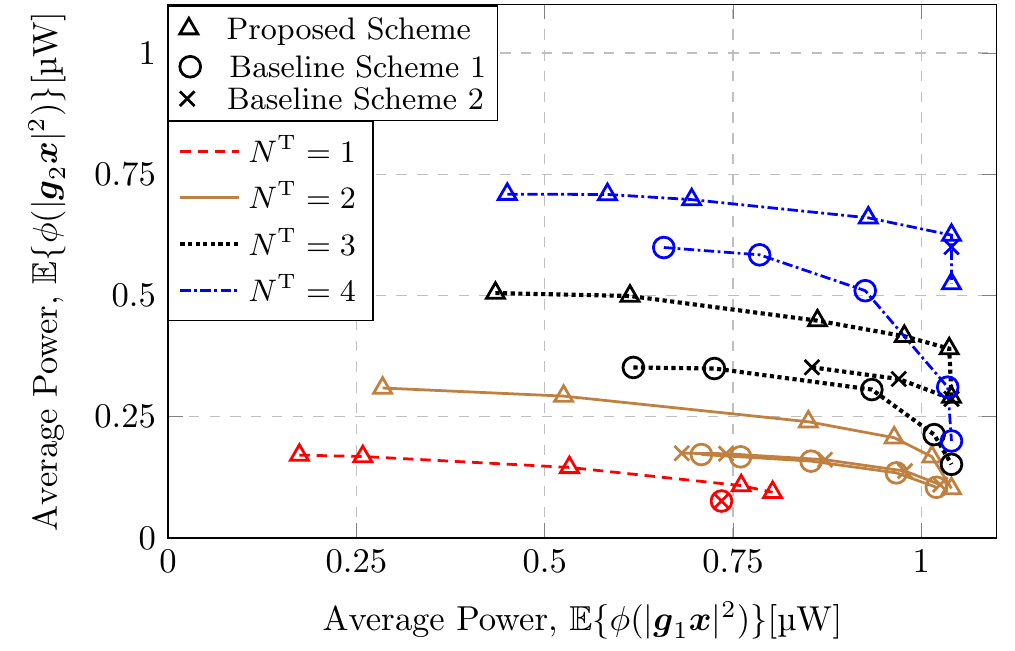} \label{Fig:Results_MU_Region_HP}}
\caption{Harvested power regions of the two-user MISO WPT system for different numbers of transmit antennas $N_\text{t}$.}
	\label{Fig:Results_MU_Region}
	\vspace*{-10pt}
\end{figure*}

In Fig.~\ref{Fig:Results_MU_Region}, we show the boundaries of the harvested power regions, whose individual points are obtained based on Proposition~\ref{Theorem:Proposition3} and Algorithms~\ref{OptimizationAlgorithmFunction} and \ref{AlgorithmGridSearch} by varying the user weights $\xi_m \in [0,1], m\in\{1,2\}$.
For the optimal pdfs $p_{\boldsymbol{x}}^*(\boldsymbol{x})$, we determine the average harvested powers at the EH nodes as $\mathbb{E}_{\boldsymbol{x}} \{\phi(| \boldsymbol{g}_m \boldsymbol{x}|^2)\}$, respectively.
As Baseline Schemes 1 and 2, we adopt the solution obtained in \cite{Zhang2013} for a linear EH model and, similar to \cite{Shen2020}, the TX design employing a single beamforming vector obtained with Algorithm~\ref{OptimizationAlgorithmFunction}, respectively.
In Fig.~\ref{Fig:Results_MU_Region_LP}, we consider a low transmit power regime characterized by a power budget of $P_x = \SI{5}{\watt}$, whereas for the results in Fig.~\ref{Fig:Results_MU_Region_HP}, we assume a high transmit power regime with $P_x = \SI{30}{\watt}$.
The distances between the EH nodes and the TX are equal to $d_1 = \SI{10}{\meter}$ and $d_2 = \SI{25}{\meter}$, respectively.

In Fig.~\ref{Fig:Results_MU_Region}, we observe that, for all considered schemes, higher values of $N_\text{t}$ and $P_x$ yield larger average harvested powers at both EH nodes.
Furthermore, the proposed scheme yields a superior performance compared to the baseline schemes.
We note that Baseline Scheme 2, where the TX is equipped with a single beamforming vector obtained with Algorithm~\ref{OptimizationAlgorithmFunction}, outperforms Baseline Scheme 1, where energy beamforming is adopted \cite{Zhang2013}.
For both transmit power regimes, we observe that for two-user SISO WPT systems, i.e., $N_\text{t} = 1$, the performances obtained with both baseline schemes are identical and do not depend on the adopted weights $\xi_1$ and $\xi_2$ since, in this case, the transmit strategy is determined by the transmit power budget $P_x$.
Moreover, in the high transmit power regime, we observe that the performance of Baseline Scheme 2 also does not depend on $\xi_m, m = \{1,2\},$ for large numbers of transmit antennas, i.e., $N_\text{t} = 4$, since, for large $N_\text{t}$, EH node $1$ is driven into saturation anyways.
However, for the other system setups, the choice of the weights $\xi_m, m\in\{1,2\},$ enables a trade-off between the powers harvested at the EH nodes, which is characterized by a convex \emph{harvested power region}. 
Furthermore, by increasing weight $\xi_m$, more power is harvested at the EH node $m$ at the expense of a reduction of the powers harvested by the other node.
Thus, by choosing the user weights, the TX can control the distribution of the harvested power among the users.
In particular, for $\xi_1=1$ (and $\xi_2 = 0$), the TX maximizes the average harvested power at the EH node $1$ and neglects EH node $2$, which may yield a substantial decrease of the power at EH node $2$, as can be observed in Fig.~\ref{Fig:Results_MU_Region_HP}.
However, since, in the high transmit power regime, EH node $1$ is almost driven into saturation for  $N_\text{t} = 4$, by decreasing $\xi_1$ (and increasing $\xi_2$), it is possible to significantly increase the harvested power at EH node $2$ without a substantial reduction of the power at user $1$.

	\section{Conclusion}
	In this paper, we studied the achievable harvested power region of two-user MISO WPT systems with non-linear EH nodes.
We characterized the points of the harvested power region by the pdfs of the transmit symbol vectors. 
The boundary of the harvested power region can be obtained by considering all the possible convex combinations of the harvested powers at the EH nodes.
In order to obtain the points on the boundary of the harvested power region, we formulated an optimization problem for the maximization of the weighted sum of the average harvested powers at the EH nodes under a constraint on the power budget of the TX.
We showed that it is optimal to employ scalar unit norm transmit symbols with arbitrary phase and two beamforming vectors.
In order to determine these beamforming vectors, we formulated a non-convex optimization problem and proposed an efficient iterative algorithm to obtain a suboptimal solution.
Our simulation results reveal that the proposed scheme significantly outperforms baseline schemes based on a linear EH model and a single beamforming vector at the TX, respectively.
Furthermore, we observed a trade-off between the individual powers harvested at the EH nodes, which is characterized by a convex harvested power region.

	\appendices
	\begin{appendices}
		% Grab forced line break - \\* - and replace with :
		\renewcommand{\thesection}{\Alph{section}}
		\renewcommand{\thesubsection}{\thesection.\arabic{subsection}}
		\renewcommand{\thesectiondis}[2]{\Alph{section}:}
		\renewcommand{\thesubsectiondis}{\thesection.\arabic{subsection}:}	
		\section{Proof of Lemma 1}
		\label{Appendix:LemmaProof}
		In the following, we prove Lemma \ref{Theorem:Corollary2}.
First, we note that since $\nu_2^*$ is the maximizer of the slope function $s(\cdot, \cdot)$ for $\nu_1 = {\nu}^*_1$, then, $\forall \nu \geq \overline{\nu}$, we have
\begin{equation}
\frac{f(\nu_2^*) - f({\nu}^*_1)}{\nu_2^* - {\nu}^*_1} \geq \frac{f(\nu) - f({\nu}^*_1)}{\nu - {\nu}^*_1}.
\label{Eqn:LemmaProof1}
\end{equation}
Then, since $\nu_1^*$ is the minimizer of $s(\cdot, \cdot)$ for $\nu_2 = \nu_2^*$, $\forall \nu \leq \overline{\nu}$, we have
\begin{align}
&\frac{f(\nu_2^*) - f(\nu_1^*)}{\nu_2^* - \nu_1^*} \leq \frac{f(\nu_2^*) - f(\nu)}{\nu_2^* - \nu}   \Longleftrightarrow \\
&f(\nu_1^*)\nu_2^* - f(\nu)(\nu_2^*-\nu_1^*) \geq f(\nu_2^*)(\nu_1^* - \nu) + f(\nu_1^*)\nu.
\label{Eqn:LemmaProof2}
\end{align} 
Next, we subtract $f(\nu_1^*)\nu_1^*$ from both sides of (\ref{Eqn:LemmaProof2}). 
This allows us to rewrite both (\ref{Eqn:LemmaProof1}) and (\ref{Eqn:LemmaProof2}) as follows:
\begin{equation}
f(\nu_1^*) - f(\nu) \geq \frac{f(\nu_2^*) - f(\nu_1^*)}{\nu_2^* - \nu_1^*} (\nu_1^* - \nu),
\label{Eqn:LemmaProof3}
\end{equation}
\noindent respectively, which, thus, holds $\forall \nu \in \mathbb{R}$.
Let us define linear function $g(\nu) = f(\nu_1^*) + \frac{f(\nu_2^*) - f(\nu_1^*)}{\nu_2^* - \nu_1^*} (\nu - \nu_1^*)$.
From (\ref{Eqn:LemmaProof3}), we have $f(\nu) \leq g(\nu)$. Thus, the objective function in problem (\ref{Eqn:GeneralOptimizationProblem}) is upper-bounded by $\mathbb{E}_\nu \{f(\nu)\} \leq \mathbb{E}_\nu \{g(\nu)\}$, which, hence, yields Lemma~\ref{Theorem:Corollary2}.% \cite{Zhang2012}.
This concludes the proof.

		\section{Proof of Proposition 1}
		\label{Appendix:Prop3Proof}
			First, we note that for any arbitrary transmit symbol $\tilde{\boldsymbol{x}}$, there is a symbol $\hat{\boldsymbol{x}}$ given by
	\begin{equation}
		\hat{\boldsymbol{x}} = \argmax_{\boldsymbol{x}} \Psi(\boldsymbol{x}) \; \subjectto \norm{\boldsymbol{x}}^2 = \norm{\tilde{\boldsymbol{x}}}^2,
	\end{equation}  
	which has the same transmit power and yields a higher or equal value of $\Psi(\boldsymbol{x})$.
	Hence, for any arbitrary distribution of transmit symbols with a point of increase $\tilde{\boldsymbol{x}}$, a larger value of $\Psi(\boldsymbol{x})$ can be obtained by removing this point and increasing the probability of $\hat{\boldsymbol{x}}$ by the corresponding value.
	
	Let us introduce now the monotonically non-decreasing function $\Phi(\nu)$ in (\ref{Eqn:MimoPropositionFunction}) that returns the largest possible value of $\Psi(\boldsymbol{x})$ if a symbol with power $\nu$ was transmitted. 
	Since (\ref{Eqn:EhUtilityFunction}) can be rewritten as $\overline{\Phi}(p_{\boldsymbol{x}}) = \mathbb{E}\{\Psi(\boldsymbol{x})\}$, the solution of (\ref{Eqn:WPT_GeneralProblem}) can be obtained by determining first the solution $p^*_\nu(\nu)$ of the following optimization problem:
	\begin{equation}
		\maximize_{p_{\nu} (\nu)} \; \mathbb{E}_\nu\{\Phi(\nu)\}\; \quad \subjectto \; \mathbb{E}_\nu\{\nu\} \leq P_x,
		\label{Eqn:MimoPropositionProofProblem}
	\end{equation}
	and then, since $p^*_\nu(\nu)$ is a discrete pdf, evaluating the optimal pdf $p^*_{\boldsymbol{x}} (\hat{\boldsymbol{x}})$ as 
	\begin{equation}
		p^*_{\boldsymbol{x}}(\hat{\boldsymbol{x}}) = p^*_{\nu}(\nu), \quad \text{where} \; \hat{\boldsymbol{x}} = \argmax_{\boldsymbol{x}: \, \norm{\boldsymbol{x}}^2 = \nu} \Psi(\boldsymbol{x}).
	\end{equation}	
	Thus, since (\ref{Eqn:MimoPropositionProofProblem}) is in the form of (\ref{Eqn:GeneralOptimizationProblem}), applying Lemma~\ref{Theorem:Corollary2}, yields Proposition~\ref{Theorem:Proposition3}.
	This concludes the proof.
	
	\end{appendices}

	\bibliographystyle{IEEEtran}
	\bibliography{Final}

\begin{thebibliography}{10}
\providecommand{\url}[1]{#1}
\csname url@samestyle\endcsname
\providecommand{\newblock}{\relax}
\providecommand{\bibinfo}[2]{#2}
\providecommand{\BIBentrySTDinterwordspacing}{\spaceskip=0pt\relax}
\providecommand{\BIBentryALTinterwordstretchfactor}{4}
\providecommand{\BIBentryALTinterwordspacing}{\spaceskip=\fontdimen2\font plus
\BIBentryALTinterwordstretchfactor\fontdimen3\font minus
  \fontdimen4\font\relax}
\providecommand{\BIBforeignlanguage}[2]{{%
\expandafter\ifx\csname l@#1\endcsname\relax
\typeout{** WARNING: IEEEtran.bst: No hyphenation pattern has been}%
\typeout{** loaded for the language `#1'. Using the pattern for}%
\typeout{** the default language instead.}%
\else
\language=\csname l@#1\endcsname
\fi
#2}}
\providecommand{\BIBdecl}{\relax}
\BIBdecl

\bibitem{Grover2010}
P.~{Grover} and A.~{Sahai}, ``{Shannon} meets {Tesla}: {Wireless} information
  and power transfer,'' in \emph{Proc. IEEE Int. Symp. Information Theory},
  Jun. 2010, pp. 2363--2367.

\bibitem{Zhang2013}
R.~{Zhang} and C.~K. {Ho}, ``{MIMO} broadcasting for simultaneous wireless
  information and power transfer,'' \emph{IEEE Trans. Wirel. Commun.}, vol.~12,
  no.~5, pp. 1989--2001, May 2013.

\bibitem{Boshkovska2015}
E.~{Boshkovska}, D.~W.~K. {Ng}, N.~{Zlatanov}, and R.~{Schober}, ``Practical
  non-linear energy harvesting model and resource allocation for {SWIPT}
  systems,'' \emph{IEEE Commun. Lett.}, vol.~19, no.~12, pp. 2082--2085, Dec.
  2015.

\bibitem{Kim2020}
J.~Kim, B.~Clerckx, and P.~D. Mitcheson, ``Signal and system design for
  wireless power transfer: Prototype, experiment and validation,'' \emph{{IEEE}
  Trans. Wirel. Commun.}, vol.~19, no.~11, pp. 7453--7469, Nov. 2020.

\bibitem{Clerckx2016a}
B.~Clerckx and E.~Bayguzina, ``Waveform design for wireless power transfer,''
  \emph{{IEEE} Trans. on Signal Process.}, vol.~64, no.~23, pp. 6313--6328,
  Dec. 2016.

\bibitem{Clerckx2018}
B.~{Clerckx}, ``Wireless information and power transfer: Nonlinearity, waveform
  design, and rate-energy tradeoff,'' \emph{IEEE Trans. Signal Process.},
  vol.~66, no.~4, pp. 847--862, Feb. 2018.

\bibitem{Huang2017}
Y.~Huang and B.~Clerckx, ``Large-scale multiantenna multisine wireless power
  transfer,'' \emph{{IEEE} Trans. on Signal Process.}, vol.~65, no.~21, pp.
  5812--5827, Nov. 2017.

\bibitem{Shen2020}
S.~Shen and B.~Clerckx, ``Beamforming optimization for {MIMO} wireless power
  transfer with nonlinear energy harvesting: {RF} combining versus {DC}
  combining,'' \emph{IEEE Trans. Wirel. Commun.}, vol.~20, no.~1, pp. 199--213,
  Jan. 2021.

\bibitem{Morsi2019}
R.~{Morsi}, V.~{Jamali}, A.~{Hagelauer}, D.~W.~K. {Ng}, and R.~{Schober},
  ``Conditional capacity and transmit signal design for {SWIPT} systems with
  multiple nonlinear energy harvesting receivers,'' \emph{IEEE Trans. Commun.},
  vol.~68, no.~1, pp. 582--601, Jan. 2020.

\bibitem{Shanin2020}
N.~Shanin, L.~Cottatellucci, and R.~Schober, ``Markov decision process based
  design of {SWIPT} systems: Non-linear {EH} circuits, memory, and impedance
  mismatch,'' \emph{{IEEE} Trans. Commun.}, vol.~69, no.~2, pp. 1259 -- 1274,
  Feb. 2021.

\bibitem{Coope2001}
I.~D. Coope and C.~J. Price, ``On the convergence of grid-based methods for
  unconstrained optimization,'' \emph{{SIAM} J. on Optim.}, vol.~11, no.~4, pp.
  859--869, Jan. 2001.

\bibitem{Ghanem2020}
W.~R. Ghanem, V.~Jamali, and R.~Schober, ``Resource allocation for secure
  multi-user downlink {MISO}-{URLLC} systems,'' in \emph{{IEEE} Int. Conf. on
  Commun. ({ICC})}, Jun. 2020.

\bibitem{Sun2017}
Y.~Sun, P.~Babu, and D.~P. Palomar, ``Majorization-minimization algorithms in
  signal processing, communications, and machine learning,'' \emph{{IEEE}
  Trans. on Signal Process.}, vol.~65, no.~3, pp. 794--816, Feb. 2017.

\bibitem{Dokov2002}
S.~P. Dokov and D.~P. Morton, \emph{Higher-Order Upper Bounds on the
  Expectation of a Convex Function}.\hskip 1em plus 0.5em minus 0.4em\relax
  Humboldt-Universit{\"a}t zu Berlin, 2002.

\bibitem{Grant2015}
M.~Grant and S.~Boyd, ``{CVX}: Matlab software for disciplined convex
  programming, version 2.0 beta (2013),'' \emph{URL: http://cvxr. com/cvx},
  2015.

\bibitem{Xu2019}
D.~Xu, X.~Yu, Y.~Sun, D.~W.~K. Ng, and R.~Schober, ``Resource allocation for
  secure {IRS}-assisted multiuser {MISO} systems,'' in \emph{{IEEE} Globecom
  Workshop}, Dec. 2019.

\bibitem{Polik2010}
I.~P{\'o}lik and T.~Terlaky, \emph{Interior point methods for nonlinear
  optimization}.\hskip 1em plus 0.5em minus 0.4em\relax Springer, 2010.

\end{thebibliography}

\end{document}